\documentclass[aps,prl,twocolumn,superscriptaddress,nofootinbib,floatfix]{revtex4-2}

\usepackage{amsmath,amssymb,bm}
\usepackage{graphicx}
\usepackage{float}
\usepackage{hyperref}
\hypersetup{hidelinks}
\newcommand{\cube}{\{0,1\}^{n}}
\newcommand{\Tr}{\operatorname{Tr}}
\newcommand{\rank}{\operatorname{rank}}
\newcommand{\ket}[1]{|#1\rangle}
\newcommand{\bra}[1]{\langle #1|}

\newtheorem{theorem}{Theorem}

\begin{document}

\title{Slepian Bounds on the Success Probability of Virtual Distillation}

\author{Masayuki Ohzeki}
\email{mohzeki@tohoku.ac.jp}
\affiliation{Graduate School of Information Sciences, Tohoku University, Sendai 980-8579, Japan}
\affiliation{Department of Physics, Institute of Science Tokyo, Meguro, Tokyo 152-8551, Japan}
\affiliation{Research and Education Institute for Semiconductors and Informatics, Kumamoto University, Kumamoto 860-8555, Japan}
\affiliation{Sigma-i Co., Ltd., Minato, Tokyo 108-0075, Japan}
\author{Sasuke Kimata}
\affiliation{Graduate School of Information Sciences, Tohoku University, Sendai 980-8579, Japan}
\author{Xinwei Lee}
\affiliation{Singapore Management University, Singapore}
\author{Hoong Chuin Lau}
\affiliation{Singapore Management University, Singapore}

\date{\today}

\begin{abstract}
Virtual distillation is a powerful near-term error-mitigation primitive, but it
is also a spectral filter: it amplifies the dominant eigenvector component
already present in the noisy density matrix.  We show that, for a finite-band
variational state, this filtering cannot create new concentration inside a set
of accepted measurement outcomes.  The asymptotic success probability after
distillation is bounded by the leading eigenvalue of a Slepian concentration
operator built from the variational Fourier band and that outcome window.
Moreover, the number of robust high-success spectral components is limited by
the Slepian active dimension.  In bit-string variational optimization,
including the quantum approximate optimization algorithm, this operator has an
explicit Krawtchouk kernel on the Boolean hypercube.  Numerical verification on
generic finite-dimensional band-window pairs confirms that virtual
distillation does not exceed the Slepian bounds.
\end{abstract}

\maketitle

\textit{Introduction.---}
Variational quantum algorithms are constrained by the finite information
content of their circuit ansatz.  In the quantum approximate optimization
algorithm (QAOA), this limitation appears as a problem-dependent relation
between circuit depth, the rank of the quantum Fisher information matrix, and
optimization success~\cite{Farhi2014,Cerezo2021,Rabinovich2026}.  That rank
counts locally independent parameter directions in the variational state.  A
useful way to express this relation is Slepian concentration: a finite Fourier
band and a target
window cannot be simultaneously resolved beyond a finite number of degrees of
freedom~\cite{Slepian1961,Slepian1978,Simons2009}.  Here a band means the
linear span of Fourier modes supplied by the ansatz, while a window means the
set of basis states counted as successful outcomes.  The Slepian eigenvalues
measure how much probability a normalized band-limited vector can place inside
that window.  Graph signal processing
uses the same principle to quantify vertex-frequency localization on finite
graphs~\cite{Tsitsvero2016,VanDeVille2017}, while variational quantum models
are known to admit finite Fourier descriptions determined by their
generators~\cite{Schuld2021}.

Error mitigation does not remove this information limit automatically.  Virtual
distillation estimates observables with respect to a normalized power of the
noisy density matrix~\cite{Huggins2021,Koczor2021}.  Its original analyses
already emphasize that the method is governed by the dominant eigenvector of
that noisy density matrix, and general bounds on quantum error mitigation show
that mitigation protocols cannot freely undo information loss without
cost~\cite{Takagi2022}.  The question addressed here is more specific: can
power filtering increase the success probability of a finite-band variational
state beyond the Slepian concentration allowed by its ansatz?  The answer is
no.  Distillation can change spectral weights among branches already present,
but it cannot create new directions selected by the Slepian eigenvectors.

\textit{Problem setting.---}
Consider an \(n\)-qubit register with computational basis
\(\{\ket{x}:x\in\cube\}\) and Hilbert-space dimension
\(\mathcal D=2^n\).  A solution window \(A\subseteq\cube\) is the set of bit
strings that are counted as successful outputs for the optimization problem.
For example, it can be the set of strings whose objective value is within a
fixed tolerance of the best known value.  The noisy output density matrix is
\begin{equation}
\rho=\sum_{a=0}^{\mathcal D-1}\lambda_a\ket{\psi_a}\bra{\psi_a}.
\end{equation}
The eigenvalues satisfy \(\lambda_a\ge0\), \(\sum_a\lambda_a=1\), and are
ordered as \(\lambda_0\ge\lambda_1\ge\cdots\).  We call each eigenvector
\(\ket{\psi_a}\) a spectral branch.  For a positive integer \(m\), the
virtual-distilled state obtained from \(m\) copies is
\begin{equation}
\sigma_m=\frac{\rho^m}{\Tr\rho^m}.
\label{eq:sigma_m}
\end{equation}
Let \(P_A\) be the computational-basis projector onto the span of
\(\{\ket{x}:x\in A\}\).  The target-window probability is the Born probability
of measuring a string in \(A\):
\begin{equation}
p_A^{(m)}=\Tr(P_A\sigma_m).
\label{eq:target_prob}
\end{equation}
If the largest eigenvalue is nondegenerate, the large-copy limit converges to
\begin{equation}
p_A^{(\infty)}=\bra{\psi_0}P_A\ket{\psi_0}.
\label{eq:asymptotic_prob}
\end{equation}
If the largest eigenvalue is degenerate, \(\sigma_m\) converges instead to the
normalized projector onto the top eigenspace.  Thus virtual distillation
selects the dominant noisy spectral branch, or the dominant spectral subspace
in the degenerate case.  The nontrivial question is how large
Eq.~(\ref{eq:asymptotic_prob}) can be when the selected branch remains inside
the finite information band of the ansatz.

Let \(\Pi_p\) be the orthogonal projector onto the Fourier band associated with
a depth-\(p\) variational ansatz.  The theorem below only requires this
projector; the later Krawtchouk construction gives an explicit Walsh-band
realization.  The Slepian concentration operator for the target window is
\begin{equation}
S_{A,p}=\Pi_pP_A\Pi_p.
\label{eq:slepian_operator}
\end{equation}
It is a positive contraction on the band, so its eigenvalues lie between zero
and one.  Its eigenvalue equation is
\begin{equation}
S_{A,p}\ket{\phi_k}=\mu_k\ket{\phi_k}.
\label{eq:slepian_eigs}
\end{equation}
We order the eigenvalues as
\(\mu_1^{(p)}\ge\mu_2^{(p)}\ge\cdots\ge0\).  The leading value
\(\mu_1^{(p)}\) is the maximum target-window probability attainable by any
normalized vector in the band.  For a threshold \(0\le\tau<1\), define the
active Slepian dimension by
\begin{equation}
D_\tau(p)=\#\{k:\mu_k^{(p)}>\tau\}.
\label{eq:dtau}
\end{equation}
It counts the number of Slepian modes whose individual window concentration is
larger than \(\tau\).  The threshold is an analysis resolution, not a circuit
or mitigation parameter.  There is no universal optimal value: a useful
\(\tau\) is set by the desired success scale or by a stable gap in the Slepian
spectrum, both of which depend on the circuit-induced band and the target
window.

\begin{theorem}[Slepian no-go]
\normalfont
Assume in-band support of the noisy density matrix,
\begin{equation}
\operatorname{supp}(\rho)\subseteq\operatorname{Ran}\Pi_p.
\label{eq:band_assumption}
\end{equation}
Here \(\operatorname{Ran}\Pi_p\) denotes the range of the band projector.
Then virtual distillation cannot increase the target-window probability beyond
the leading Slepian concentration:
\begin{equation}
p_A^{(m)}\le \mu_1^{(p)}.
\label{eq:mu_bound}
\end{equation}
Let \(Q_\tau\) be the projector onto the Slepian modes with
\(\mu_k^{(p)}>\tau\), namely the span of the first \(D_\tau(p)\) modes.  The
stronger active-sector bound is
\begin{equation}
p_A^{(m)}\le \tau+(1-\tau)\Tr(Q_\tau\sigma_m).
\label{eq:active_bound}
\end{equation}
Consequently, achieving a desired success probability \(p_A^{(m)}\ge s>\tau\)
requires
\begin{equation}
\Tr(Q_\tau\sigma_m)\ge \frac{s-\tau}{1-\tau}.
\label{eq:active_necessary}
\end{equation}
The required weight lies in the active Slepian subspace, whose dimension is
\(D_\tau(p)\).
\end{theorem}

\textit{Proof.---}
The support assumption implies
\begin{equation}
\sigma_m=\Pi_p\sigma_m\Pi_p.
\end{equation}
Therefore
\begin{equation}
p_A^{(m)}=\Tr(S_{A,p}\sigma_m).
\end{equation}
Use the spectral resolution of the Slepian operator inside the band,
\begin{equation}
S_{A,p}=\sum_k\mu_k^{(p)}\ket{\phi_k}\bra{\phi_k}.
\end{equation}
Then
\begin{equation}
p_A^{(m)}=\sum_k\mu_k^{(p)}\bra{\phi_k}\sigma_m\ket{\phi_k}.
\end{equation}
Writing \(\alpha_k=\bra{\phi_k}\sigma_m\ket{\phi_k}\), each \(\alpha_k\) is
nonnegative and \(\sum_k\alpha_k=1\), because \(\sigma_m\) is normalized and
supported in the band.  Since \(\mu_k^{(p)}\le\mu_1^{(p)}\), the weighted sum is
at most \(\mu_1^{(p)}\), proving Eq.~(\ref{eq:mu_bound}).  For the
active-sector statement, define \(\mathcal I_\tau=\{k:\mu_k^{(p)}>\tau\}\).
Split the weighted sum into active and inactive Slepian modes:
\begin{equation*}
p_A^{(m)}
=
\sum_{k\in\mathcal I_\tau}\mu_k^{(p)}\alpha_k
+
\sum_{k\notin\mathcal I_\tau}\mu_k^{(p)}\alpha_k.
\end{equation*}
For \(k\in\mathcal I_\tau\), the eigenvalue satisfies \(\mu_k^{(p)}\le1\);
for \(k\notin\mathcal I_\tau\), it satisfies \(\mu_k^{(p)}\le\tau\).
Therefore
\begin{equation}
p_A^{(m)}
\le
\sum_{k\in\mathcal I_\tau}\alpha_k
+\tau\sum_{k\notin\mathcal I_\tau}\alpha_k.
\end{equation}
Using \(\sum_k\alpha_k=1\), this becomes
\(\tau+(1-\tau)\sum_{k\in\mathcal I_\tau}\alpha_k\).  This is exactly
Eq.~(\ref{eq:active_bound}), because
\(\sum_{k\in\mathcal I_\tau}\alpha_k=\Tr(Q_\tau\sigma_m)\).  Rearranging it
gives the necessary active-sector weight for any target \(s>\tau\). \(\square\)

If noise moves the state outside the variational band, the theorem becomes a
leakage statement rather than a violation.  Let \(I\) denote the identity on
the full Hilbert space.  Define the distilled out-of-band weight as
\begin{equation}
\eta_m=\Tr[(I-\Pi_p)\sigma_m].
\label{eq:leakage}
\end{equation}
The quantity \(\eta_m\) is the probability weight of the distilled state
outside the ansatz band.  For any positive \(\sigma_m\),
\begin{equation}
p_A^{(m)}
\le
\mu_1^{(p)}(1-\eta_m)+\eta_m+2\sqrt{\eta_m(1-\eta_m)}.
\label{eq:leakage_bound}
\end{equation}
This loose bound follows by decomposing \(P_A\) into in-band, out-of-band, and
cross blocks and bounding the cross term by the Cauchy--Schwarz inequality.
It is useful because any apparent excess over \(\mu_1^{(p)}\) must be
accompanied by nonzero \(\eta_m\).
Moreover, the leakage itself is spectrally filtered:
\begin{equation}
\eta_m=
\frac{\sum_a\lambda_a^m\bra{\psi_a}(I-\Pi_p)\ket{\psi_a}}
{\sum_a\lambda_a^m}.
\label{eq:eta_spectrum}
\end{equation}
Thus subleading out-of-band branches are suppressed as \(m\) increases if their
eigenvalues are smaller than the leading in-band eigenvalue.  A persistent
excess beyond the in-band Slepian bound therefore means that the dominant
spectral branch selected by virtual distillation is itself out of band.

\textit{Krawtchouk realization.---}
Many variational-circuit applications to discrete optimization return a bit
string after measurement, and success is judged by whether that string lies in
a chosen subset of \(\cube\).  The quantum approximate optimization algorithm
is the canonical example: its cost Hamiltonian is diagonal in the
computational basis, so bit strings are the solution coordinate.  The standard
transverse-field mixer is diagonalized by the Walsh-Hadamard transform, which
gives the frequency coordinate.  We use a Walsh degree cutoff \(w\) as an
explicit finite-band realization of \(\Pi_p\):
\begin{equation}
\Pi_w=H^{\otimes n}P_{\le w}H^{\otimes n}.
\end{equation}
Here \(H\) is the one-qubit Hadamard, and \(P_{\le w}\) keeps Walsh strings of
Hamming weight at most \(w\).  For a depth-\(p\) circuit with a known frequency
set \(\Omega_p\), replace \(\Pi_w\) by the projector onto \(\Omega_p\).  The
band-window operator is
\begin{equation}
S_{A,w}=\Pi_wP_A\Pi_w.
\end{equation}
Equivalently, after restricting rows and columns to \(x,y\in A\), the matrix
entries depend only on the Hamming distance \(d=d(x,y)\):
\begin{equation}
K_w(d)=2^{-n}\sum_{q=0}^{w}K_q^{(n)}(d).
\label{eq:kraw_kernel}
\end{equation}
The Krawtchouk polynomial is
\begin{equation}
K_q^{(n)}(d)=
\sum_{j=0}^{q}(-1)^j
\binom{d}{j}
\binom{n-d}{q-j}.
\label{eq:kraw_poly}
\end{equation}
The binomial coefficients vanish outside their natural ranges.  The associated
Shannon number is
\begin{equation}
N_{A,w}=\Tr S_{A,w}.
\label{eq:shannon_trace}
\end{equation}
It is the trace of the concentration operator and plays the role of the
expected number of well-concentrated modes.  For the Boolean hypercube,
\begin{equation}
N_{A,w}=|A|2^{-n}\sum_{q=0}^{w}\binom{n}{q}.
\label{eq:shannon}
\end{equation}
This is the number of phase-space cells jointly available to the variational
band and the target window.

The binary case therefore obeys the same virtual-distillation bound with
explicitly computable Slepian data.  Let \(\mu_k^{(w)}\) be the eigenvalues of
\(S_{A,w}\), and let \(Q_{\tau,w}\) project onto the modes with
\(\mu_k^{(w)}>\tau\).  If the noisy state is Walsh-band limited,
\begin{equation}
\operatorname{supp}(\rho)\subseteq\operatorname{Ran}\Pi_w,
\end{equation}
then Theorem~1 gives
\begin{equation}
p_A^{(m)}\le \mu_1^{(w)}.
\end{equation}
It also gives
\begin{equation}
p_A^{(m)}\le \tau+(1-\tau)\Tr(Q_{\tau,w}\sigma_m).
\end{equation}
Thus a bit-string variational circuit cannot exceed the Slepian concentration
bound determined by its Walsh band and target window; the Krawtchouk kernel
only makes that bound explicit.

\textit{Numerical verification.---}
We verify the no-go inequalities without using any optimization instance.
For each trial we choose a finite Hilbert space of dimension \(64\), a random
finite-rank band projector \(\Pi\), and a random computational-basis outcome
window \(A\).  From \(S_A=\Pi P_A\Pi\) we compute the Slepian eigenvalues and
the active projector \(Q_\tau\).  The plotted data use \(\tau=0.30\) as a
representative threshold; the inequalities themselves hold for every
\(0\le\tau<1\).  We next draw random density
matrices supported entirely in \(\operatorname{Ran}\Pi\), apply the virtual
distillation map of Eq.~(\ref{eq:sigma_m}), and compare the resulting target
probabilities with the analytic bounds.  This construction tests the theorem
as a statement about arbitrary band-window pairs, not as a property of a
particular cost function.

\begin{figure}[H]
\includegraphics[width=\columnwidth]{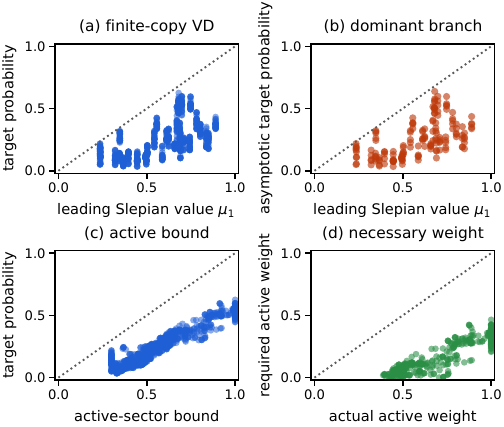}
\caption{
Generic numerical verification of the Slepian bounds.  Each point comes from a
random band projector, a random outcome window, and a random in-band density
matrix.  Panel (a) compares finite-copy virtual-distilled probabilities with
the leading Slepian value \(\mu_1\).  Panel (b) compares the asymptotic
dominant-branch probability with \(\mu_1\).  Panel (c) tests the active-sector
upper bound of Eq.~(\ref{eq:active_bound}).  Panel (d) tests the necessary
active-weight condition of Eq.~(\ref{eq:active_necessary}) by plotting the
required active weight against the actual active weight.  Dotted diagonals mark
equality, and satisfying points lie below the diagonal in all panels.
}
\label{fig:bounds}
\end{figure}

Figure~\ref{fig:bounds} shows no violations of the leading bound, the
asymptotic bound, or the active-sector bound.  The fourth panel gives the same
statement in the necessary-condition form: when a distilled state has target
probability above \(\tau\), its weight in the active Slepian subspace is at
least the amount required by Eq.~(\ref{eq:active_necessary}).

\textit{Discussion.---}
The result separates what is known about virtual distillation from what the
Slepian analysis adds.  Prior work established that virtual distillation is a
power method on the noisy spectrum and that its usefulness depends on the
dominant eigenvector~\cite{Huggins2021,Koczor2021}.  The present statement adds
a variational information constraint: if that dominant spectral branch remains
inside the finite Fourier band of the ansatz, its target-window probability is
bounded by the Slepian concentration eigenvalue, and robust high-success
branches must live in a subspace of dimension \(D_\tau(p)\).

The present theorem does not by itself predict the observed
overparameterization depth or optimized success probability of a finite QAOA
instance.  Its role is structural: it identifies the target-window information
capacity available inside a specified variational band.  Numerical studies of
quantum Fisher information matrix rank, noise-induced tangent directions, and
optimization success probe how particular ansatz trajectories populate this
capacity.

Equation~(\ref{eq:active_bound}) is restrictive when the active Slepian
subspace is small compared with the full band.  Let
\begin{equation}
r_p=\rank\Pi_p.
\end{equation}
For a Haar-generic normalized vector \(\ket{\psi}\) in the band,
\begin{equation}
\mathbb E\,\bra{\psi}Q_\tau\ket{\psi}=\frac{D_\tau(p)}{r_p}.
\end{equation}
The corresponding typical estimate from Eq.~(\ref{eq:active_bound}) is
\begin{equation}
\mathbb E\,p_A\le \tau+(1-\tau)\frac{D_\tau(p)}{r_p}.
\end{equation}
Thus the bound is small when \(D_\tau(p)/r_p\) is small and \(\tau\) is chosen
below the target success scale.  Large target probability requires the
dominant spectral branch to be unusually aligned with the active Slepian
subspace.  Virtual distillation can amplify such a branch if it is already
present in the noisy spectrum, but it does not create the alignment.

This clarifies the relation to overparameterization.  The Slepian active
dimension is a target-conditioned finite-band capacity, whereas quantum Fisher
information matrix rank counts the available tangent dimension without
reference to the target window.  Virtual distillation can purify one branch
inside the active Slepian sector, but it cannot enlarge that sector.  When
mitigation appears to exceed the in-band concentration bound, the diagnosis is
not that distillation has created new variational directions; it is that noise
supplied an out-of-band dominant branch.  Slepian analysis therefore supplies
the capacity statement missing from the usual spectral description of virtual
distillation: the branch selected by the power method can have high
optimization success only through its preexisting concentration inside the
finite-band active Slepian sector.

\begin{acknowledgments}
We received financial support from the Cross-ministerial Strategic Innovation
Promotion Program (SIP) of the Cabinet Office (No. 23836436).
\end{acknowledgments}

\end{document}